\documentclass[11pt,letterpaper]{article}
\pdfoutput=1

\usepackage[dvipsnames,usenames]{xcolor}
\usepackage[colorlinks=true,urlcolor=Blue,citecolor=Green,linkcolor=BrickRed]{hyperref}
\usepackage{amsmath,amsthm,amssymb,amsfonts}
\usepackage[utf8]{inputenc}
\usepackage{todonotes}
\usepackage{xspace}
\usepackage{multicol}
\usepackage{thmtools,thm-restate}
\usepackage{authblk}
\usepackage{hyperref}
\usepackage{multirow}
\usepackage{algorithm}
\usepackage{algorithmicx}
\usepackage[noend]{algpseudocode}
\usepackage{tikz}
\usepackage[noadjust]{cite}
\usepackage{cellspace}
\usepackage{adjustbox}
\usepackage{amsmath}
\usepackage{mathtools}
\usepackage{cite}

\usepackage[left=1in,right=1in,top=1in,bottom=1in]{geometry}

\newtheorem{theorem}{Theorem}[section]
\newtheorem{lemma}[theorem]{Lemma}

\newcommand{\Oh}{\mathcal{O}}

\makeatletter
\newcommand\thankssymb[1]{\textsuperscript{\@fnsymbol{#1}}}
\makeatother

\begin{document}

\title{Lower Bounds for the Number of Repetitions in 2D Strings}
\date{}

\author[1]{Paweł Gawrychowski\thanks{Partially supported by the Bekker programme of the Polish National Agency for Academic
Exchange (PPN/BEK/2020/1/00444).}}
\author[2]{Samah Ghazawi\thanks{Partially supported by the Israel Science Foundation grant 1475/18, and Grant No. 2018141 from the United States-Israel
Binational Science Foundation (BSF).}}
\author[1]{Gad M. Landau\thankssymb{2}}

\affil[1]{Institute of Computer Science, University of Wroc\l{}aw, Poland}
\affil[2]{Department of Computer Science, University of Haifa, Israel}
\affil[3]{NYU Tandon School of Engineering, New York University, Brooklyn, NY, USA}

\maketitle
\begin{abstract}
    A two-dimensional string is simply a two-dimensional array.
    We continue the study of the combinatorial properties of repetitions in such strings over the binary alphabet,
    namely the number of distinct tandems, distinct quartics, and runs.
    First, we construct an infinite family of $n\times n$ 2D strings with $\Omega(n^{3})$ distinct tandems.
    Second, we construct an infinite family of $n\times n$ 2D strings with $\Omega(n^{2}\log n)$ distinct quartics.
    Third, we construct an infinite family of $n\times n$ 2D strings with $\Omega(n^{2}\log n)$ runs.
    This resolves an open question of Charalampopoulos, Radoszewski, Rytter, Waleń, and Zuba
    [ESA 2020], who asked if the number of distinct quartics and runs in an $n\times n$ 2D string is $\Oh(n^{2})$.
\end{abstract}
\section{Introduction}

The study of repetitions in strings goes back at least to the work of Thue from 1906~\cite{thue}, who constructed an infinite
square-free word over the ternary alphabet. Since then, multiple definitions of repetitions have been proposed
and studied, with the basic question being focused on analyzing how many such repetitions a string of length $n$ can contain.
The most natural definition is perhaps that of palindromes, which are fragments that read the same either from left to right or right to left.
Of course, any fragment of the string $\texttt{a}^{n}$ is a palindrome, therefore we would like to count distinct
palindromes. An elegant folklore argument shows that this is at most $n+1$ for any string of length $n$~\cite{DROUBAY2001539},
which is attained by $\texttt{a}^{n}$.

Another natural definition is that of squares, which are fragments of the form $xx$, where $x$ is a string.
Again, because of the string $\texttt{a}^{n}$ we would like to count distinct squares. 
Using a combinatorial result of Crochemore and Rytter~\cite{Crochemore95}, Fraenkel and Simpson~\cite{FS} proved that a string of length $n$
contains at most $2n$ distinct squares (see also~\cite{Ilie05} for a simpler proof, and~\cite{Ilie07} for an upper bound
of $2n-\Theta(\log n)$). They also provided an infinite family of strings of length $n$ with $n-o(n)$ distinct squares.
It is conjectured that the right upper bound is actually $n$, however so far we only know that it is at most
$11/6n$~\cite{DEZA}.
Interestingly, a proof of the conjecture for the binary alphabet would imply it for any alphabet~\cite{ManeaS15}.

Perhaps a bit less natural, but with multiple interesting applications, is the definition of runs. A run is
a maximal periodic fragment that is at least twice as long as its smallest period. Roughly speaking, runs
capture all the repetitive structure of a string, making them particularly useful when constructing algorithms~\cite{CROCHEMORE2014}.
A well-known result by Kolpakov and Kucherov~\cite{Kolpakov} is that a string of length $n$ contains
$\Oh(n)$ runs; they conjectured that it is actually at most $n$. After a series of improvements~\cite{rytter2,PSS,ilie2},
with the help of an extensive computer search the upper bound was decreased to $1.029n$~\cite{ilie3,G}.
Finally, in a remarkable breakthrough Bannai et al.~\cite{Bannai17} confirmed the conjecture.
On the lower bound side, we current know an infinite family of strings with at least $0.944575712n$ runs~\cite{FY,MKIBS,simpson}.
Interestingly, better bounds are known for the binary alphabet~\cite{FHIL}.

Given that we seem to have a reasonably good understanding of repetitions in strings,
it is natural to consider repetitions in more complex structures, such as circular strings~\cite{AG,CurrieF02,Simpson14}
or trees~\cite{CrochemoreIKKRRTW12,GawrychowskiKRW15,KociumakaRRW17}.
In this paper, we are interested in repetitions in 2D strings. Naturally, algorithms operating
on 2D strings can be used for image processing, and combinatorial properties of such strings
can be used for designing efficient pattern matching algorithms~\cite{Amir,Amir1,Amir2,Amir3,Cole}.
Therefore, we would like to fully understand what is a repetition in a 2D string, and what is the
combinatorial structure of such repetition.

Apostolico and Brimkov~\cite{AB} introduced the notions of tandems and quartics in 2D strings.
Intuitively, a tandem consists of two occurrences of the same block $W$ arranged in a $1\times 2$ or $2\times 1$ pattern,
while a quartic consists of 4 occurrences of the same block $W$ arranged in a $2\times 2$ pattern. They
considered tandems and quartics with a primitive $W$, meaning that it cannot be partitioned
into multiple occurrences of the same $W'$ (called primitively rooted in the subsequent work~\cite{CRRWZ}),
and obtained asymptotically tight bounds of $\Theta(n^{2}\log^{2}n)$ and $\Theta(n^{2}\log n)$
for the number of such tandems and quartics in an $n\times n$ 2D string, respectively. 
Both tandems and quartics should be seen as an attempt to extend the notion of squares in a 1D
string to 2D strings, and thus the natural next step is to consider distinct tandems and quartics (without restricting $W$ to be primitive).
Very recently, Charalampopoulos et al.~\cite{CRRWZ} studied the number of distinct tandems and quartics in an $n\times n$
2D string.
For distinct tandems, they showed a tight bound of $\Theta(n^{3})$ with the construction used in the lower
bound using an alphabet of size $n$.
For distinct quartics, they showed an upper bound of $\Oh(n^{2}\log^{2}n)$ and conjectured
that it is always $\Oh(n^{2})$, similarly to the number of distinct squares in a 1D string of length $n$
being $\Oh(n)$.

Amir et al.~\cite{ALMS,ALMS1} introduced the notion of runs in 2D strings.
Intuitively, a 2D run is a maximal subarray that is both horizontally and vertically periodic; we defer
a formal definition to the next section.
They proved that an $n\times n$ 2D string contains $\Oh(n^{3})$ runs, showing an infinite
family of $n\times n$ 2D strings with $\Omega(n^{2})$ runs.
Later, Charalampopoulos et al.~\cite{CRRWZ} significantly improved on this upper bound,
showing that an $n\times n$ 2D string contains $\Oh(n^{2}\log^{2}n)$ runs,
and conjectured that it is always $\Oh(n^{2})$, similarly to the number of runs in a 1D string of length $n$
being $\Oh(n)$.

\paragraph{Our results.}
In this paper, we consider 2D strings and obtain improved lower bounds for the number of distinct tandems, distinct quartics,
and runs. We start with the number of distinct tandems and extend the lower bound of Charalampopoulos et al.~\cite{CRRWZ}
over the binary alphabet in Section~\ref{sec:tandems} by showing the following.

\begin{restatable}{theorem}{tandem}
\label{thm:tandem}
There exists an infinite family of $n\times n$ 2D strings over the binary alphabet containing
$\Omega(n^{3})$ distinct tandems.
\end{restatable}

Then, we move to the number of distinct quartics in Section~\ref{sec:quartics} and the number of runs in Section~\ref{sec:runs}, and show the following.

\begin{restatable}{theorem}{quartics}
\label{thm:quartics}
There exists an infinite family of $n\times n$ 2D strings over the binary alphabet containing
$\Omega(n^{2}\log n)$ distinct quartics.
\end{restatable}

\begin{restatable}{theorem}{runs}
\label{thm:runs}
There exists an infinite family of $n\times n$ 2D strings over the binary alphabet containing
$\Omega(n^{2}\log n)$ runs.
\end{restatable}

\noindent By the above theorem, the algorithm of Amir et al.~\cite{ALMS} for locating all 2D runs
in $\Oh(n^{2}\log n+\textsf{output})$ time is worst-case optimal.

Our constructions exhibit a qualitative difference between distinct squares and runs in 1D strings
and distinct quartics and runs in 2D strings. The number of the former is linear in the size of the input,
while the number of the latter, surprisingly, is superlinear.

\paragraph{Our techniques.}
For distinct tandems, our construction is similar to that of~\cite{CRRWZ}, except that we use distinct characters
only in two columns. This allows us to replace them by their binary expansions, with some extra care
as to not lose any counted tandems.

For both distinct quartics and runs, we proceed recursively,
constructing larger and larger 2D strings $A_{i}$ starting from the initial 2D string $A_{1}$.
The high-level ideas behind both constructions are different, though.

For distinct quartics, our high-level idea is to consider subarrays with $\Theta(\log n)$ different
aspect ratios. For each such aspect ratio, we create $\Omega(n^{2})$ distinct quartics, for an $n\times n$ array.
Each step of the recursion corresponds to a different aspect ratio and creates multiple new special characters,
as to make the new quartics distinct; later we show how to implement this kind of approach with the binary alphabet.

For runs, we directly proceed with a construction for the binary alphabet, and build on the insight used
by Charalampopoulos et al.~\cite{CRRWZ} to show that the same quartic can be induced by $\Theta(n^{2})$ runs.
Each step of the recursion corresponds to runs with asymptotically the same size. 
This needs to be carefully analyzed in order to lower bound the overall number of runs. 

\section{Preliminaries}

Let $\Sigma$ be a fixed finite alphabet. A two-dimensional string (or 2D string, for short) over $\Sigma$
is an $m\times n$ array $A[0..m-1][0..n-1]$ with $m$ rows and $n$ columns, with every cell $A[i][j]$ containing
an element of $\Sigma$. Furthermore, we use $\epsilon$ to denote an empty 2D string.
A subarray $A[x_1..x_2][y_1..y_2]$ of $A[0..m-1][0..n-1]$ is an $(x_{2}-x_{1}+1) \times (y_{2}-y_{1}+1)$ array
consisting of cells $A[i][j]$ with $i\in [x_{1},x_{2}], j\in [y_{1},y_{2}]$.

We consider three notions of repetitions in 2D strings.
\begin{description}
  \item[Tandem.] A subarray $T$ of $A$ is a tandem if it consists of $2\times1$ (or $1\times2$) subarrays $W\neq \epsilon$.
   Two tandems $T=\begin{array}{|c|c|}
    \hline
   W & W \\
    \hline
  \end{array}$ and $T'=\begin{array}{|c|c|}
    \hline
   W' & W' \\
    \hline
  \end{array}$ are distinct when $W\neq W'$. 
    \item[Quartic.] A subarray $Q$ of $A$ is a quartic if it consists of $2\times2$ subarrays $W \neq \epsilon$.
Two quartics $Q=\begin{array}{|c|c|}
    \hline
   W & W \\
   \hline
   W & W\\
    \hline
  \end{array}$ and $Q'=\begin{array}{|c|c|}
    \hline
   W' & W' \\
   \hline
   W' & W'\\
    \hline
  \end{array}$ are distinct when $W\neq W'$.
  \item[Run.] Consider an $r\times c$ subarray $R$ of $A$. We define a positive integer $p$ to be its horizontal period
  if the $i^\text{th}$ column of $R$ is equal to the $(i+p)^\text{th}$ column of $R$, for all $i=1,2,\ldots,c-p$.
  The horizontal period of $R$ is its smallest horizontal period, and we say that $R$ is $h$-periodic when its horizontal period
  is at most $c/2$. Similarly, we define a vertical period, the vertical period, and a $v$-periodic subarray. 
  An $h$-periodic and $v$-periodic $R$ is called a run when extending $R$ in any direction would result in a subarray
  with a larger horizontal or vertical period. Informally, for such $R$ there exists a subarray $W$ such that we can represent
  $R$ as follows, with at least two repetitions of $W$ in both directions, and we cannot extend $R$ in any direction
  while maintaining this property.
  
  \begin{center}
    $R=\begin{array}{|c|c|c|c|}
    \hline
   W &\ldots& W &W'\\
   \hline
   \ldots&\ldots&\ldots&\ldots\\
   \hline
   W &\ldots& W&W'\\
    \hline
    W'' &\ldots& W''&W'''\\
    \hline
  \end{array}$
  \end{center}
  Where
  $W = \begin{array}{|c|c|}\hline
  W'&U\\
  \hline
  \end{array}$, $W = \begin{array}{|c|}\hline
  W''\\
  \hline
  U'\\
  \hline
  \end{array}$ and $W=\begin{array}{|c|c|}
    \hline
   W''' & V \\
   \hline
   V' & V'''\\
    \hline
  \end{array}$,
    and any of the subarrays $W',W'',W''',U,U',V,V'$ and $V'''$ may be $\epsilon$.
\end{description}

\section{Distinct Tandems}
\label{sec:tandems}
In this section, we show how to construct an $n\times n$ array $A$ over the binary alphabet with $\Theta(n^{3})$ distinct tandems,
for any $\ell\geq1$, where $n=3\cdot 2^{\ell}+2\ell$.
The $i^{\text{th}}$ row of $A$ is divided into 5 parts, see Figure~\ref{fig:rowi}.
The first, third, and fifth part each consists of $2^{\ell}$ cells, each containing the binary representation of 1.
The second and fourth part each consists of $\ell$ cells that contains the binary representation of the number $i-1$.
Hence, all rows of $A$ are different, see Figure~\ref{fig:tandems}.

\begin{figure}[ht]
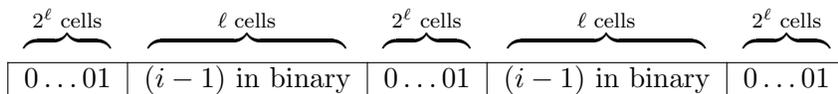

\begin{adjustbox}{max width=\textwidth}
\begin{tabular}{ *{5}{c} }
    $\overbrace{\hphantom{0 \ldots 01}}^{2^\ell\text{ cells}}$&$\overbrace{\hphantom{(i-1) \text{in binary}}}^{\ell\text{ cells}}$&$\overbrace{\hphantom{0 \ldots 01}}^{2^\ell\text{ cells}}$&$\overbrace{\hphantom{(i-1) \text{in binary}}}^{\ell\text{ cells}}$&$\overbrace{\hphantom{0 \ldots 01}}^{2^\ell\text{ cells}}$\\
  \cline{1-5}
    \multicolumn{1}{|c}{$0 \ldots 01$} & 
    \multicolumn{1}{|c}{$(i-1)$ in binary} & 
    \multicolumn{1}{|c}{$0 \ldots 01$} & 
    \multicolumn{1}{|c}{$(i-1)$ in binary} & 
    \multicolumn{1}{|c|}{$0 \ldots 01$} \\
  \cline{1-5}
\end{tabular}
\end{adjustbox}
\centering
\caption{The $i^{\text{th}}$ row of $A$.}
\label{fig:rowi}
\end{figure}

\tandem*

\begin{proof}
To lower bound the number of tandems in $A$, consider any $1\leq i\leq j \leq n$ and $k\in \{1,2,\ldots,2^{\ell}\}$.
Then, let $T$ be the subarray of width $2(2^{\ell}+\ell)$ starting in the $i^{\text{th}}$ row and ending in the $j^{\text{th}}$ row
with the top left cell of $T$ being $A[i][k]$.
We claim that for each choice of $i,j,k$ we obtain a distinct tandem, making 
the number of distinct tandems in $A$ at least $n^2\cdot 2^{\ell}=\Omega(n^{3})$.
It is clear that each such $T$ is a tandem. To prove that all of them are distinct, consider any such $T$.
The position of the leftmost 1 in its top row
allows us to recover the value of $k$. Then, the next $\ell$ cells contain the binary expansion
of $(i-1)$, so we can recover $i$. Finally, the height of $T$ together with $i$ allows us to recover
$j$. Thus, we can uniquely recover $i,j,k$ from $T$, and all such tandems are distinct.
\end{proof}

\begin{figure}[ht]
\includegraphics[width=0.4\textwidth]{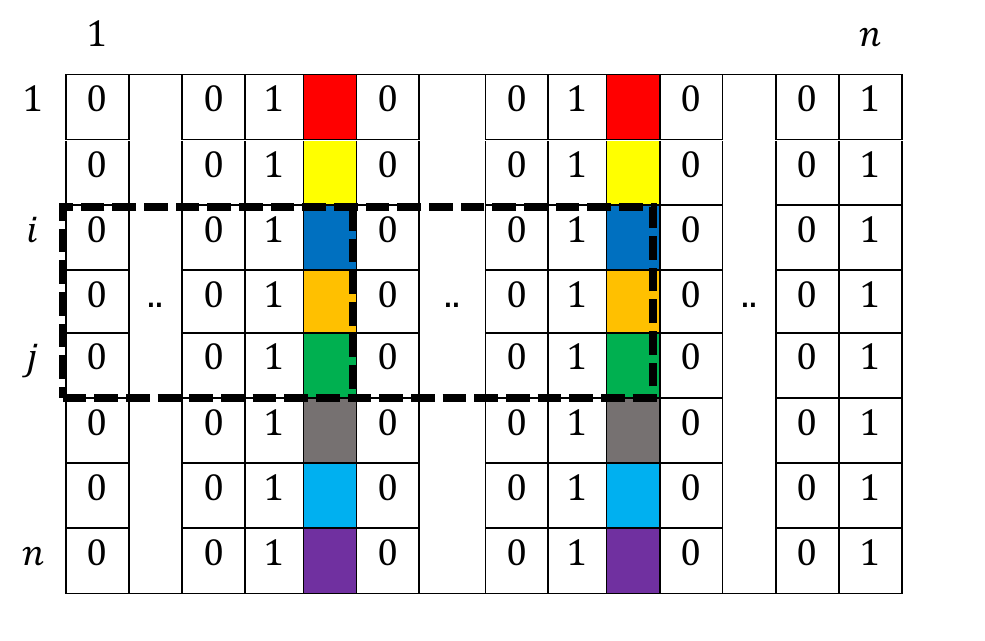}
\centering
\caption{Array $A$, where each color corresponds to the binary representation of the row number.
The black borders correspond to the leftmost tandem of height $j-i+1$ and width $2(2^{\ell}+\ell)$; 
by shifting it to the right we obtain distinct tandems.}
\label{fig:tandems}
\end{figure}
\section{Distinct Quartics}
\label{sec:quartics}
In this section, we show how to construct an $n \times n$ array $A_\ell$ with $\Omega(n^2\log n)$ distinct quartics,
for any $\ell \geq 1$, where $n=3^{\ell}-1$.
The construction is recursive, that is, we construct a series of arrays $A_1, A_2, \ldots, A_\ell$, with $A_{i}$ being defined
using $A_{i-1}$. 
The number of columns of each array $A_{i}$ is the same and equal to $n$. The number of rows
is increasing, starting with $2$ rows in $A_1$ and ending with $n$ rows in the final array $A_\ell$.
We provide the details of the construction in the next subsection, then analyze the number of distinct quartics
in $A_{\ell}$ in the subsequent subsection.
Finally, in the last subsection we show how to use $A_{\ell}$ to obtain an $n' \times n'$ array $A'_{\ell}$ over the binary
alphabet with $\Omega(n'^2\log n')$ distinct quartics.

\subsection{Construction}
First, we provide array $A_1$ of size $2\times n$ with $0$ in all but $4$ cells, namely, cells
$A_1[1][\frac{n-2}{3}+1]$, $A_1[2][\frac{n-2}{3}+1]$, $A_1[1][\frac{2n+2}{3}]$ and $A_1[2][\frac{2n+2}{3}]$ containing
the same special character. In particular, we are dividing the columns into $3$ equal parts,
see Figure~\ref{fig:A1}.

\begin{figure}[ht]
\includegraphics[width=0.3\textwidth]{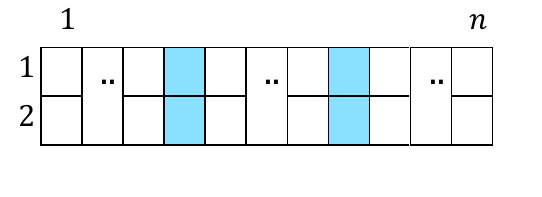}
\centering
\caption{Array $A_1$, where white cells contain $0$ and blue cells contain the same special character.
}
\label{fig:A1}
\end{figure}

Second, we describe the general construction of an $M_{i}\times n$ array $A_{i}$, for $i\geq 2$.
We maintain the invariant that the columns of $A_{i}$ are partitioned into $3^{i}$ maximal ranges of $N_{i}$ columns
consisting of only 0s and separated with single columns, i.e., $N_{1}=\frac{(n-2)}{3}$.
To obtain $A_{i}$, we first vertically concatenate 3 copies of $A_{i}$, using different special
characters in each copy, while adding a single separating row between the copies.
Thus, $M_{i}=3M_{i-1}+2$. Initially, each separating row consists of only 0s. 
For each maximal range of columns in $A_{i-1}$ that consists of only 0s, we proceed as follows.
We further partition the columns of the range into 3 sub-ranges of $\frac{N_{i-1}-2}{3}$ columns, separated by single
columns. We create a new special character and insert its four copies at the intersection of each column separating
the sub-ranges and each separating row.
Overall, we create $3^{i-1}$ new special characters.
See Figure~\ref{fig:A2} for an illustration with $i=2$ and
Figure~\ref{fig:A3} for an example with $n=26$ and $A_{3}$ being the final array.

\begin{figure}[ht]
\includegraphics[width=0.8\textwidth]{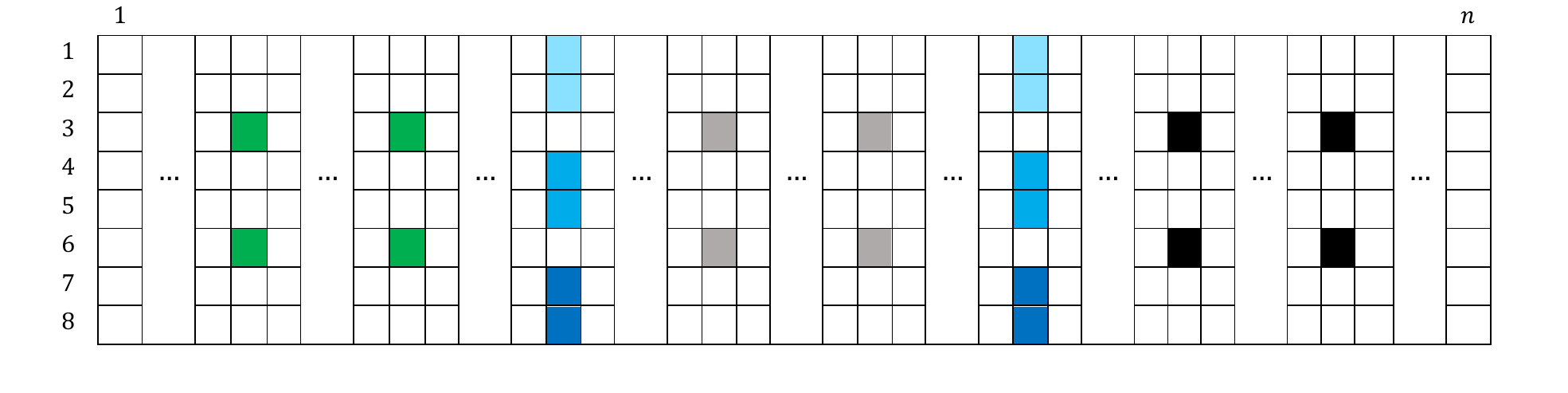}
\centering
\caption{Array $A_2$, where rows 1-2 are the first copy of $A_1$, rows 4-5 are the second copy of $A_1$, and
rows 7-8 are the third copy of $A_1$. Rows $3$ and $6$ are the separating rows. Each color corresponds to a
different special character.}
\label{fig:A2}
\end{figure}

\begin{figure}[ht]
\includegraphics[scale=0.3]{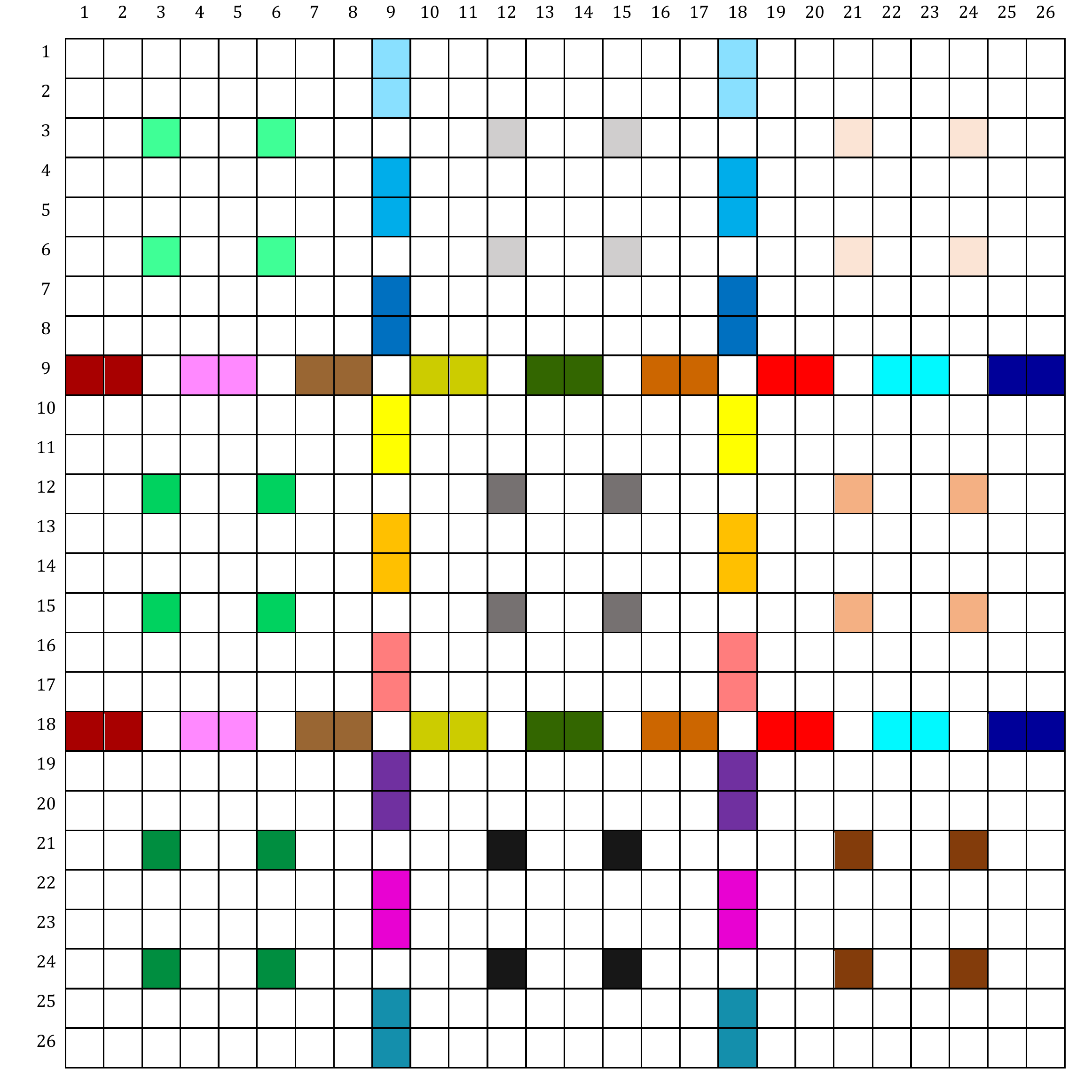}
\centering
\caption{Array $A_3$ includes $3$ copies of $A_2$ in rows 1-8, 10-17 and 19-26.
The separating rows are $9$ and $18$. Note that $A_3$ is the final array for $n=26$.
Additionally, each color represent a different special character, in total $27$ special characters are used in $A_3$.}
\label{fig:A3}
\end{figure}

\subsection{Analysis}
\label{sec:FEq}

Before we move to counting distinct quartics in each $A_{i}$, we recall that the number of columns in each
$A_{i}$ is the same and equal to $n$, while the number of rows $M_{i}$ is described by the recurrence
$M_{1}=2$ and $M_{i}=3M_{i-1}+2$ for $i\geq 2$, hence $M_{i}=3^{i}-1$. The size $N_{i}$ of each maximal range of columns
consisting only of 0s is described by the recurrence $N_1 = \frac{n-2}{3}$ and $N_i = \frac{N_{i-1}-2}{3}$ for $i\geq 2$,
hence $N_i = \frac{n+1}{3^{i}} - 1$.
By setting $n=3^{\ell}-1$ we guarantee that all these numbers are integers.

We now analyze the number of distinct quartics in each $A_{i}$. 
We will be only counting some of them, and
denote by $Q_{i}$ the distinct quartics counted in the following argument that, similarly to the construction,
considers first $i=1$ and then the general case.

For $i=1$, we count distinct quartics that contain special characters. Each of them is of width $2(\frac{n-2}{3}+1)$
and height 2. There are $Q_{1}=\frac{n-2}{3}+1=\frac{n+1}{3}$ such quartics and all of them are distinct, see Figure~\ref{fig:A1q}.

\begin{figure}[ht]
\includegraphics[width=0.3\textwidth]{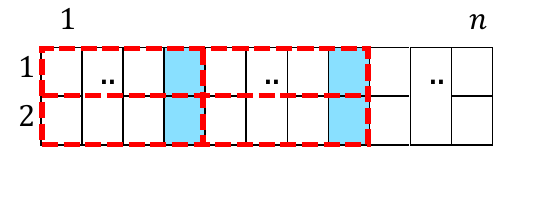}
\centering
\caption{Array $A_1$, where the red border corresponds to the leftmost quartic that contains the special character cells, and by shifting it to the right
we obtain distinct quartics.
}
\label{fig:A1q}
\end{figure}

For the general case of $i\geq 2$, we consider two groups of distinct quartics.
The first group consists of distinct quartics contained in the copies of $A_{i-1}$. 
For each of the $3^{i-1}$ maximal range of $N_{i-1}$ columns of $A_{i-1}$ consisting of 0s,
the second group consists of all possible $(2M_{i-1}+2)\times \frac{2N_{i-1}+1}{3}$ subarrays
contained in that range.
For each such range, we have $\frac{N_{i-1}-2}{3}+1$ possible horizontal shifts and $M_{i-1}+1$ possible vertical shifts
and for each of them we obtain a distinct quartic containing the new special character created for the range.
As we use different special characters in every copy of $A_{i-1}$ and, for every range, in the separating rows of $A_{i}$,
overall we have at least $Q_{i}=3Q_{i-1}+ 3^{i-1}(\frac{N_{i-1}-2}{3} + 1)(M_{i-1} + 1)$ distinct
quartics. See Figure~\ref{fig:A2q2} for an illustration with $i=2$.

\begin{figure}[ht]
\includegraphics[width=0.8\textwidth]{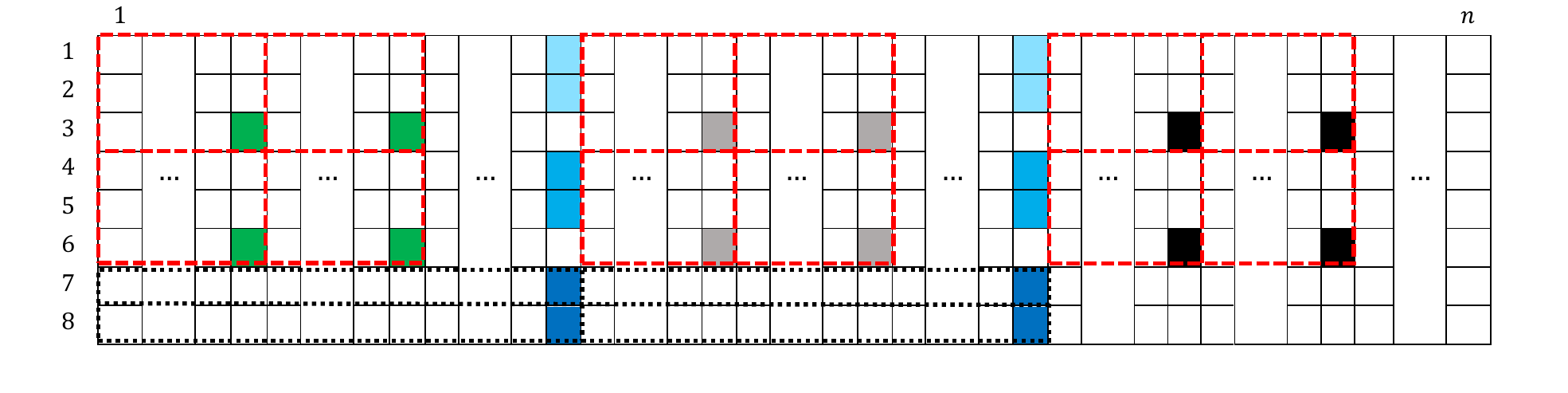}
\centering
\caption{Array $A_2$, where the black border corresponds to the
leftmost quartic in the third copy of $A_{1}$, and by shifting it to the right we obtain distinct quartics.
The red borders correspond to the leftmost quartic in each of the 3 maximal
ranges of columns of $A_{1}$ consisting of 0s; by shifting each of them to the right and down we obtain distinct quartics.}
\label{fig:A2q2}
\end{figure}

Substituting the formulas for $N_{i-1}$ and $M_{i-1}$, we conclude that
$Q_{1}=\frac{n+1}{3}$ and $Q_{i}=3Q_{i-1} + 3^{i-2}(n+1)$ for $i \geq 2$.
Unwinding the recurrence, we obtain that $Q_{i}=3^{i-1}Q_{1} +(i-1)3^{i-2}(n+1)=3^{i-1}\frac{n+1}{3} +(i-1)3^{i-2}(n+1)$.
Therefore, $Q_i = 3^{i-2}i(n+1)$.

\begin{theorem}
There exists an infinite family of $n\times n$ 2D strings containing $\Omega(n^{2}\log n)$ distinct quartics.
\end{theorem}

\begin{proof}
For each $\ell \geq 1$, we take $n=3^{\ell}-1$ and define arrays $A_{1},A_{2},\ldots,A_{\ell}$ as described above.
The final array $A_{\ell}$ consists of $M_{\ell}=n$ rows and $n$ columns, and contains at least
$Q_{\ell}=3^{\ell-2}\ell(n+1)$ distinct quartics, which is $\Omega(n^{2}\log n)$.
\end{proof}

While we were not concerned with the size of the alphabet in this construction, observe that the number of distinct special
characters $S_{i}$ in $A_{i}$ is described by the recurrences $S_{1}=1$ and $S_i = 3S_{i-1} + 3^{i-1}$ for $i\geq 2$.
This is because we are using new special characters in each copy of $A_{i-1}$ and adding $3^{i-1}$ new special characters
to divide the maximal ranges of $N_{i-1}$ columns into $3$ parts. Therefore, the size of the alphabet used
to construct $A_{\ell}$ is $S_{\ell}= \frac{(n+1)\log (n+1)}{3}+1$.

\subsection{Reducing Alphabet}
In this subsection, we show how to modify the array $A_{\ell}$ to obtain an array $A'_{\ell}$ over the binary alphabet, for any $\ell\geq 4$.
Informally speaking, we will replace each special character by a small gadget encoding its binary representation,
and then carefully revise the parameters of the new construction, particularly the number of distinct quartics.

Let $\Sigma=\{1,2,\ldots,\sigma\}$ be the alphabet used to construct $A_{\ell}$, where $\sigma=\frac{(n+1)\log (n+1)}{3}+1$.
We define arrays $B_1,B_2,\ldots,B_{\sigma}$ of the same size $ k\times k$, where $ k = \sqrt{\log \sigma}+2$.
The first row and column of every array $B_{c}$ contain only 0s, while the remaining cells of the last row and column
contain only 1s. The concatenation of cells from the middle of $B_{c}$ (without the first and last row and column),
in the left-right top-bottom order, should be equal to the binary representation of $c$.
Now, we construct the array $A'_{\ell}$ from the array $A_{\ell}$
by repeating the recursive construction of arrays $A_{1},A_{2},\ldots,A_{\ell}$,
but replacing a cell containing the character $c$ with the array $B_{c}$. We denote the resulting
arrays $A'_{1},A'_{2},\ldots,A'_{\ell}$.

We now set $n'=n\cdot  k$. Each of the arrays $A'_{i}$ consists of $n'$ columns
and $M_{i}\cdot k$ rows, so
the final array, $A'_\ell$, is of size $n'\times n'$.
We now analyze the number of distinct quartics in $A'_{i}$. 
This will be done similarly as it was for $A_{i}$,
but we must be more careful about arguing quartics as being distinct, because we no longer have
multiple distinct special characters. We first argue that, for all sufficiently wide and tall subarrays
of $R$, the horizontal and vertical shifts are uniquely defined modulo $k$.

\begin{lemma}
\label{lem:shift}
Consider a subarray $R=A'_{i}[x_{1}..x_{2}][y_{1}..y_{2}]$ with width and height at least $ k$.
Then $(x_{1}\bmod  k)$ and $(y_{1}\bmod k)$ can be recovered from $R$.
\end{lemma}

\begin{proof}
We only analyze how to recover $(y_{1}\bmod k)$, recovering $(x_{1}\bmod  k)$ is symmetric.
By construction of $B_{1},B_{2},\ldots,B_{\sigma}$, every $ k^{\text{th}}$ row of $A'_{i}$ consists of only 0s,
while in every other row there is at least one 1 in every block of $ k$ cells. Therefore,
because the width of $R$ is at least $ k$, a row of $R$ consists of 0s if and only if
it is aligned with a row of $A'_{i}$ that consists of 0s. 
Because the height of $R$ is at least
$ k$ such a row surely exists and allows us to recover $(y_{1}\bmod  k)$.
\end{proof}

We argue that the number of distinct quartics in $A'_{1}$ is at least $Q'_{1}=\frac{n\cdot k-2 k}{3} + 1$.
To show this, we consider subarrays spanning the whole height of $A'_{1}$ and of width
$2(\frac{n\cdot k-2 k}{3}+ k)$. There are $\frac{n\cdot k-2 k}{3} + 1$ such subarrays
and each of them is a quartic that fully contains some $B_{c}$. Furthermore, subarrays starting in columns with different
remainders modulo $ k$ are distinct by Lemma~\ref{lem:shift}. Subarrays
starting in columns with the same remainder modulo $ k$ are also distinct, as in such
a case we can recover the special character from $B_{c}$ fully contained in the subarray.

For the general case, we claim that the number of distinct quartics in $A'_{i}$ is at least
$Q'_{i}=3Q'_{i-1} + 3^{i-1}(\frac{N_{i-1}\cdot k-2 k}{3} + 1)(M_{i-1}\cdot k + 1)$
for $i\geq 2$. The argument proceeds as for $A_{i}$; however, we must argue that the counted quartics are
all distinct. By construction, each of them fully contains some $B_{c}$. Thus,
quartics starting in columns with different remainders modulo $ k$ (and also in rows with different
remainders modulo $ k$) are distinct by Lemma~\ref{lem:shift}. Now consider all
counted quartics starting in columns with remainder $y$ modulo $ k$ and rows
with remainder $x$ modulo $ k$. For each of them, we can recover the special
character from $B_{c}$ fully contained in the quartic, so all of them are distinct.

Finally, we lower bound and solve the recurrence for $Q'_{i}$ as follows.
\begin{align*}
Q'_{i}  &= 3Q'_{i-1} + 3^{i-1}(\frac{N_{i-1}\cdot k-2 k}{3} + 1)(M_{i-1}\cdot k + 1) \\
&> 3Q'_{i-1} + 3^{i-2}\cdot  k^{2}  (N_{i-1}-2)M_{i-1}\\
&= 3Q'_{i-1} + 3^{i-2}\cdot  k^{2}  (\frac{n+1}{3^{i-1}}-3)(3^{i-1}-1)\\
&= 3Q'_{i-1} + 3^{i-2}\cdot  k^{2}  (n-3^{i})\frac{3^{i-1}-1}{3^{i-1}} \\
&> 3Q'_{i-1} + 3^{i-3}\cdot  k^{2}  (n-3^{i}) & \MoveEqLeft \text {using } i\geq 2. \\
\end{align*}
Unwinding the recurrence, we obtain that $Q'_{i}  > \sum_{j=2}^{i} 3^{i-j} \cdot 3^{j-3}\cdot  k^{2} (n-3^{j}) > 3^{i-3}\cdot  k^{2}  ( (i-1)n - \frac{3^{i}}{2})$.

\quartics*
\begin{proof}
For each $\ell \geq 4$, we take $n=3^{\ell}-1$ and define arrays $A'_{1},A'_{2},\ldots,A'_{\ell}$ as described above.
The final array $A'_{\ell}$ is over the binary alphabet by construction, consists of $n'$ rows and $n'$ columns, where $n'=n\cdot  k$,
and contains at least $Q'_{\ell}$ distinct quartics. Finally,
\begin{align*}
Q'_{\ell} &> 3^{\ell-3}\cdot  k^{2}\cdot ((\ell-1)n-\frac{3^{\ell}}{2})\\
&=3^{\ell-3}\cdot  k^{2}\cdot ((\ell-1)(3^{\ell}-1)-\frac{3^{\ell}}{2}) \\
&> 3^{\ell-3}\cdot  k^{2}\cdot (\ell-2)(3^{\ell}-1) & \MoveEqLeft  \text{ because } \frac{3^{\ell}}{2} < 3^{\ell}-1  \\
&\geq 3^{\ell-3}\cdot  k^{2}\cdot \frac{\ell}{2} \cdot (3^{\ell}-1) & \MoveEqLeft  \text{ because } \ell-2 \geq \frac{\ell}{2}.
\end{align*}
Therefore, the number of runs in $A'_{\ell}$ is $\Omega(3^{2\ell}\cdot  k^{2}\cdot \ell)=\Omega(n'^{2}\log n')$.
\end{proof}
\section{Runs}
\label{sec:runs}
In this section, we show how to construct an $n\times n$ array $A_{\ell}$ with $\Omega(n^{2}\log n)$
runs, for any $\ell\geq 2$, where $n=2\cdot4^{\ell}$. As in the previous section, the construction is recursive, i.e.,
we construct a series of arrays $A_{1},A_{2},\ldots,A_{\ell}$, with $A_{i}$ being defined using $A_{i-1}$.
Both the number of rows and columns in $A_{i}$ is equal to $2\cdot4^{i}$, starting with 8 rows
and columns in $A_{1}$. We describe the construction in the next subsection, then analyze the number
of runs in $A_{\ell}$ in the subsequent subsection.

\subsection{Construction}
First, we provide array $A_1$ of size $8\times8$ with 1s in the cells
$A_1[1][2]$, $A_1[2][1]$, $A_1[7][8]$ and $A_1[8][7]$, and 0s in the other cells, see Figure~\ref{fig:rA12} (left).

\begin{figure}[b]
\begin{minipage}[c]{0.49\textwidth}
\centering
\includegraphics[width=0.25\textwidth]{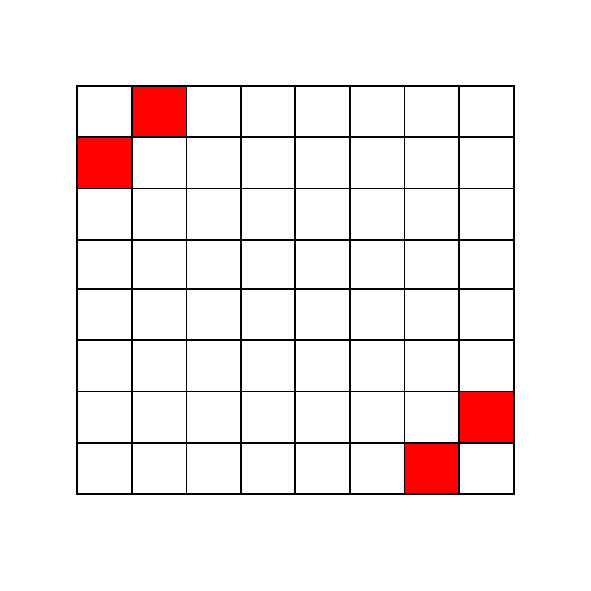}
\end{minipage}
\begin{minipage}[c]{0.49\textwidth}
\centering
\includegraphics[width=0.8\textwidth]{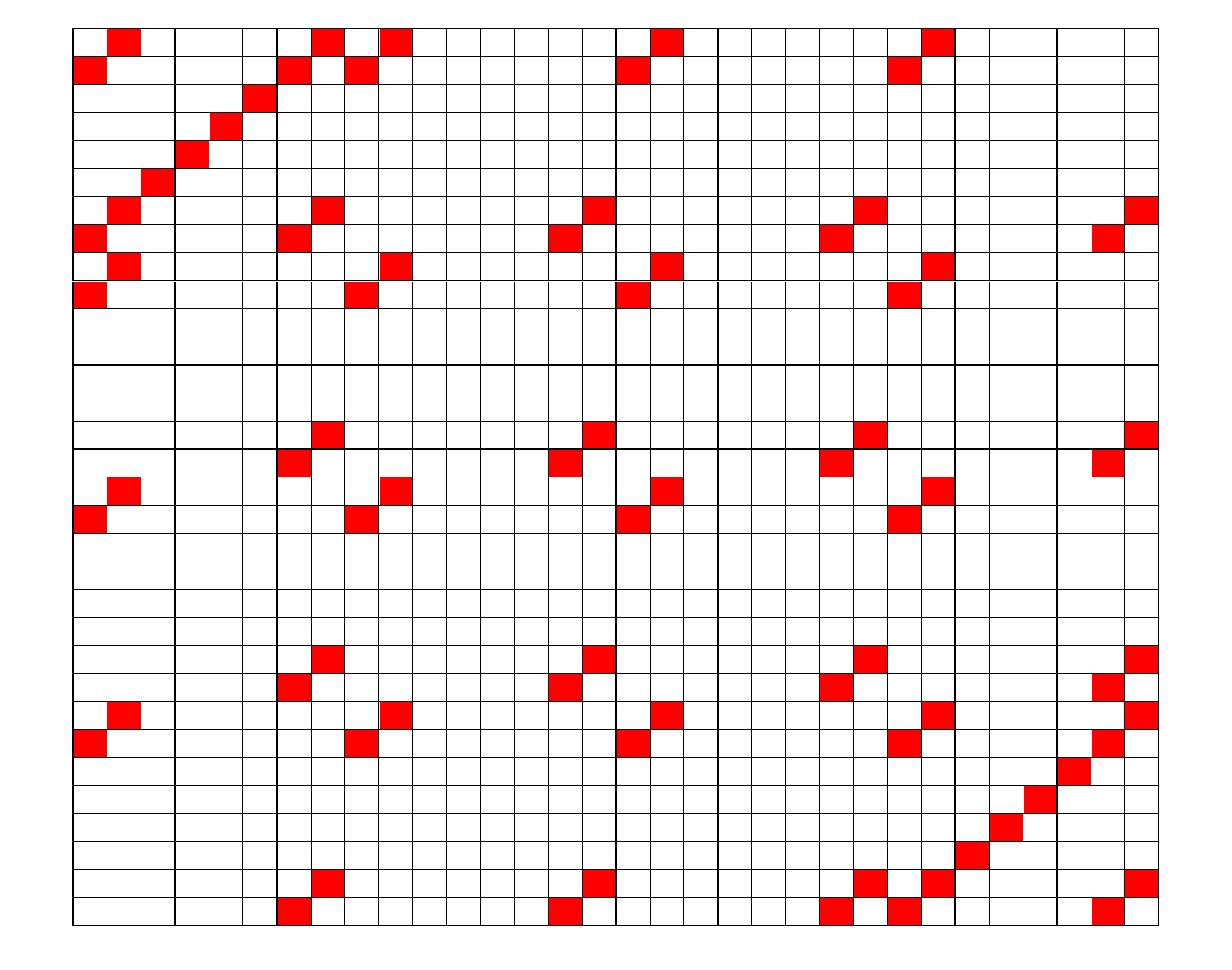}
\end{minipage}
\caption{Left: array $A_1$, where red cells contain 1s and white cells contain 0s. Right: Array $A_2$ consists of $14$ copies of $A_1$ and $2$ copies of $A'_1$. Red cells include 1s and white cells include 0s. Red is used to fill the antidiagonals of $A'_1$.}
\label{fig:rA12}
\end{figure}

Second, we obtain array $A_{i}$ by concatenating $4 \times 4$ copies of array $A_{i-1}$ while
using 1s to fill the antidiagonals in the
upper left and bottom right copy of $A_{i-1}$, with $A'_{i-1}$ denoting such modified copy of $A_{i-1}$, see Figure~\ref{fig:rA12} (right)
for an illustration with $i=2$ and
Figure~\ref{fig:runsA3} for an example with $n=128$ and $A_{3}$ being the final array.

\begin{figure}[ht]
\includegraphics[scale=0.7]{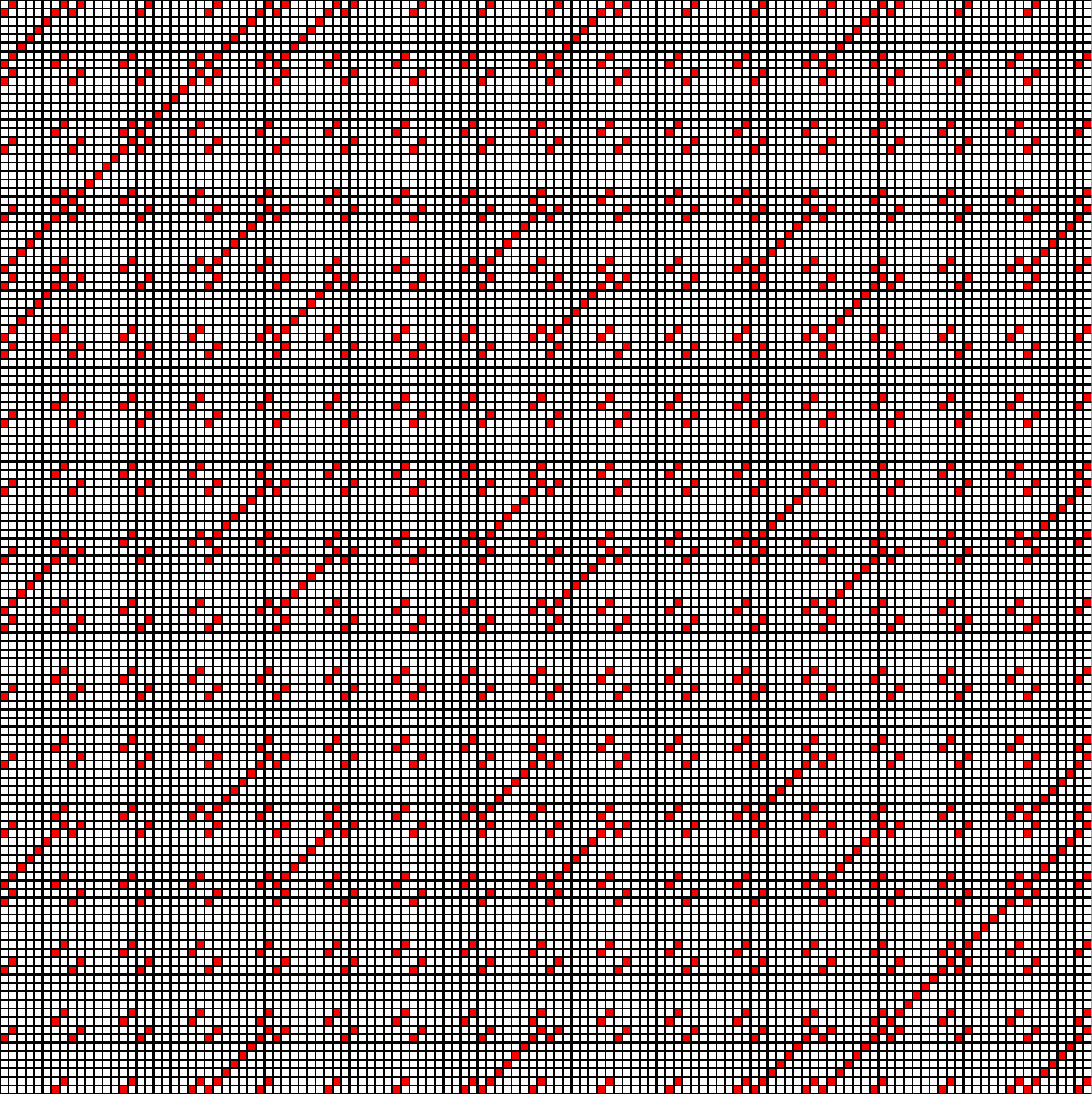}
\centering
\caption{Array $A_3$ includes $14$ copies of $A_2$, $2$ copies of $A'_2$, $216$ copies of $A_1$ and $40$ copies of $A'_1$. Note that, $A_3$ is the final array for $n=128$.}
\label{fig:runsA3}
\end{figure}

The intuition behind the recursive construction is to duplicate the runs obtained in the previous arrays.
For example, the array $A_1$ produces one run that does not touch the boundaries. This is duplicated
$14$ times in $A_2$, hence, $A_1$ contributes $14$ runs to the total number of runs produced by $A_2$.
Moreover, the intuition behind filling the antidiagonals is to produce new runs such that the number of the new runs is equal to
the size of the array up to some constant. As an example, $A_2$ produces $7^2$ new runs between the antidiagonals of
$A'_1$ such that the upper left and the bottom right corners of each run touch exactly two cells of the antidiagonals of the
two copies of $A'_1$. See Figure~\ref{fig:brA2} for an illustration.
Therefore, overall the number of runs produced by $A_2$ is $14+7^2=63$.
The general case is analyzed in detail
in the next subsection.

\begin{figure}[ht]
\includegraphics[scale=0.3]{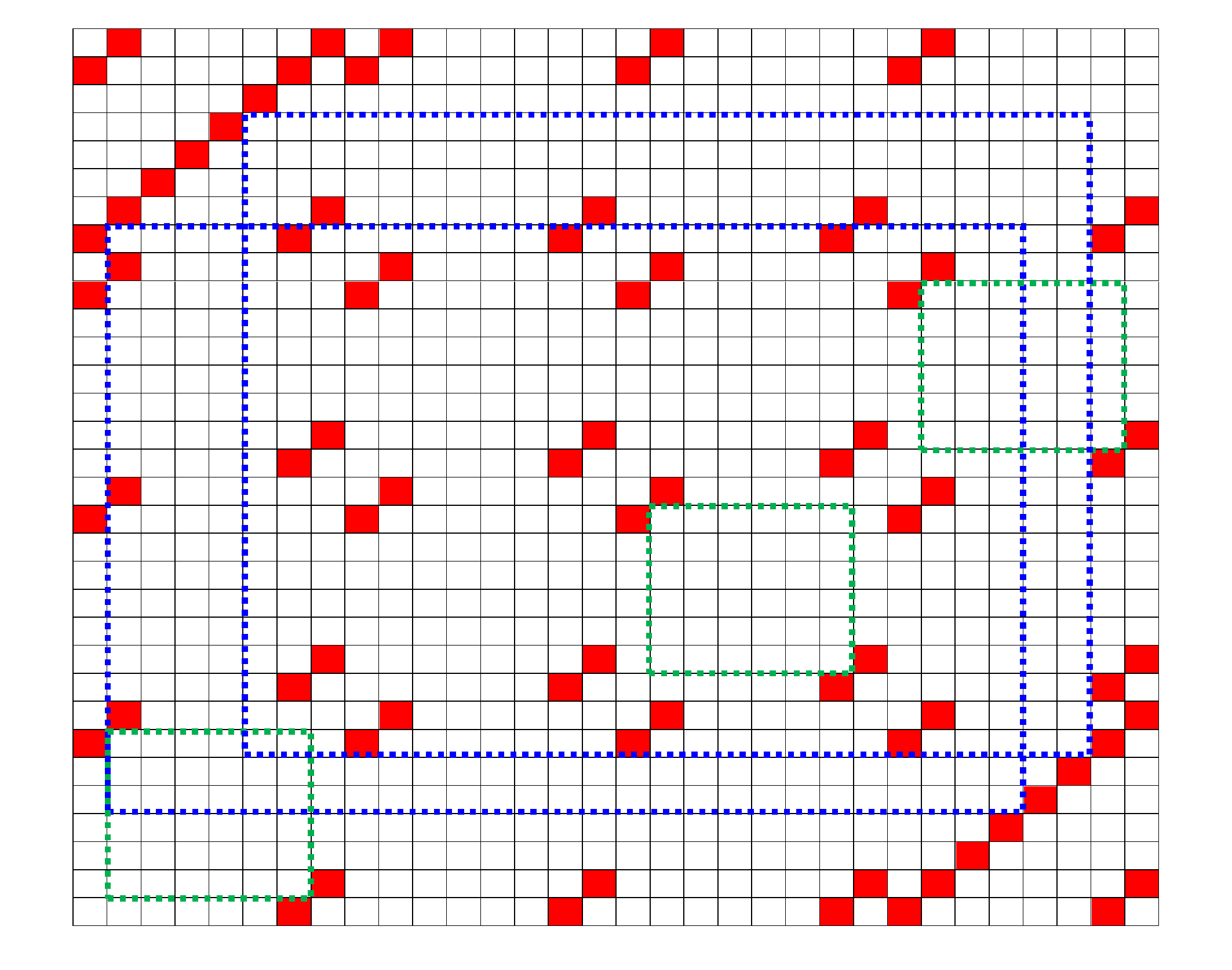}
\centering
\caption{Array $A_2$, where the blue borders correspond to new runs produced by $A_2$. The green borders correspond to runs previously produced by $A_1$.}
\label{fig:brA2}
\end{figure}

\subsection{Analysis}

The number $N_{i}$ of rows and columns in $A_i$ is described by the recurrence $N_{1}=8$ and $N_i = 4N_{i-1}$ for $i\geq 2$,
so $N_i = 2\cdot4^{i}$. By straightforward induction, the antidiagonal of every $A_{i}$ is filled with 0s.

We analyze the number of runs in $A_{i}$. First, we have $R_{i}$ new runs not contained in any of the copies of $A_{i-1}$
or $A'_{i-1}$  such that the upper left and the bottom right corners touch exactly two cells of the antidiagonals of the
two copies of $A'_{i-1}$.

\begin{lemma}
\label{lem:diagonal}
$R_i = (2\cdot 4^{i-1}-1)^2=\frac{16^i}{4}-4^i+1$.
\end{lemma}

\begin{proof}
Consider any subarray $R$ of $A_i$ with the upper left and the bottom right corners touching exactly two cells of the antidiagonals
of the two copies $A'_{i-1}$.
It is easy to verify that $N_{i-1}$ is a horizontal and a vertical period of $R$. Therefore, $R$ is $h$-periodic and $v$-periodic.
Now consider extending $R$ in any direction, say by one column to the left. Then the topmost cell of the new
column would contain a 1 from the antidiagonal of $A'_{i-1}$. For the horizontal period of the extended array to remain
$N_{i-1}$ we would need a 1 in the corresponding cell of the antidiagonal
of $A_{i-1}$, but that cell contains a 0, a contradiction. Therefore, any such $R$ is a run.
The number of such subarrays is $(N_{i-1}-1)^{2}=(2\cdot 4^{i-1}-1)^{2}=\frac{16^i}{4}-4^i+1$,
because we have $(N_{i-1}-1)$ possibilities for choosing the upper left and bottom right corner.
\end{proof}

Second, we have the runs contained in the 14 copies of $A_{i-1}$, hence $A_{i-1}$ contributes $14\cdot R_{i-1}$ to the total number of runs in $A_i$. Moreover, whenever $A_{i-1}$ contains a copy of $A_j$, for some $j < i-1$, all new runs of $A_j$ are preserved in $A_{i-1}$ and consequently in $A_i$. Additionally, we have the two copies of $A'_{i-1}$. Because we have filled their antidiagonals with 1s, we lose some of the runs. However, whenever $A'_{i-1}$ contains a copy of $A_{j}$ that does not intersect the antidiagonal, for some $j<i-1$, all new runs of $A_{j}$ are preserved in $A'_{i-1}$ and consequently in $A_{i}$. For example, each copy of $A_{i-1}$ contains $14$ copies of $A_{i-2}$ and each copy of $A'_{i-1}$ contains $10$ copies of $A_{i-2}$ (5 above and 5 below the antidiagonal). Hence, $A_i$ contains $14\cdot14+2\cdot10=216$ copies of $A_{i-2}$, thus $A_{i-2}$ contributes $216\cdot R_{i-2}$ to the total number of runs in $A_i$.
Therefore, in order to count the total number of runs in the final array $A_\ell$, we need to analyze how many copies of $A_{i}$ are in $A_{\ell}$, for $1\leq i\leq\ell$.

Let $X_{i}$ denote the number of copies of $A_{i}$ in $A_{\ell}$, and $Y_{i}$ denote the number of copies of $A'_{i}$ in $A_{\ell}$.
By construction, $A_{i}$ consists of 14 copies of $A_{i-1}$ and 2 copies of $A'_{i-1}$.
Similarly, $A'_{i}$ consists of 10 copies of $A_{i-1}$ (5 above and 5 below the antidiagonal)
and 6 copies of $A'_{i-1}$ (4 intersecting the antidiagonal and the top left and bottom right copy).
Consequently, we obtain the recurrences $X_\ell=1$ and $X_i=14X_{i+1}+10Y_{i+1}$ for $i<\ell$,
$Y_{\ell}=0$ and $Y_i=6Y_{i+1}+2X_{i+1}$ for $i<\ell$. Instead of solving the recurrences,
we show the following.

\begin{lemma}
\label{lem:approx}
$X_{i} \geq \frac{5}{6}16^{\ell-i}$
\end{lemma}

\begin{proof}
We first observe that $X_i+Y_i=16(X_{i+1}+Y_{i+1})$, as $A_{i+1}$ consists of the $4 \times 4$ smaller subarrays,
each of them being $A_{i}$ or $A'_{i}$. By unwinding the recurrence, $X_i+Y_i=16^{\ell-i}(X_\ell+Y_\ell) = 16^{\ell-i}$.
Furthermore, we argue that $X_{i} \ge 5Y_{i}$ for every $i<\ell$. This is proved by induction on $i$:
\begin{description}
\item[$i=\ell-1$] $X_{\ell-1}=14X_\ell+10Y_\ell=14\geq 5Y_{\ell-1}=5(6Y_\ell+2X_\ell)=10$.
\item[$i<\ell-1$] Assuming that $X_{i+1} \geq 5Y_{i+1}$, we write $X_{i}=14X_{i+1}+10Y_{i+1} \geq 10X_{i+1}+30Y_{i+1}$
and $5Y_{i}=30Y_{i+1}+10X_{i+1}$, so $X_{i}\geq 5Y_{i}$.
\end{description}
Therefore, $16^{\ell-i} = X_{i}+Y_{i} \leq X_{i} + \frac{X_{i}}{5}$, so $X_{i} \geq \frac{5}{6} 16^{\ell-i}$.
\end{proof}

As explained earlier, whenever a copy of $A_{i}$ occurs in $A_{\ell}$, all of its new runs contribute to the total
number of runs in $A_{\ell}$. Therefore, the total number of runs in $A_{\ell}$ is at least $\sum_{i=1}^{\ell} X_i \cdot R_i$.

\runs*

\begin{proof}
For each $\ell \geq 2$, we take $n=2\cdot 4^{\ell}$ and construct the arrays
$A_{1},A_{2},\ldots,A_{\ell}$ as described above. The final array $A_{\ell}$ is over the binary alphabet by construction, consists of $n$
rows and columns and contains at least $\sum_{i=1}^{\ell} X_i \cdot R_i$ runs. By
Lemma~\ref{lem:diagonal} and~\ref{lem:approx}, this is at least
\begin{align*}
\sum_{i=1}^{\ell} \frac{5}{6}16^{\ell-i} (2\cdot 4^{i-1}-1)^{2} &= \frac{5}{6}16^{\ell}\sum_{i=1}^{\ell}16^{-i}(\frac{16^i}{4}-4^i+1)\\
&=\frac{5}{6}16^{\ell}\sum_{i=1}^{\ell}(\frac{1}{4}-\frac{1}{4^i}+\frac{1}{16^i})\\
&=\frac{5}{6}16^{\ell}(\frac{\ell}{4}+\frac{1}{3\cdot4^\ell}-\frac{1}{15\cdot16^\ell}-\frac{4}{15})\\
&=\frac{5}{6\cdot4}\ell\cdot 16^{\ell}+\frac{5}{6\cdot3}4^{\ell}-\frac{1}{6\cdot3}-\frac{5}{6\cdot3}16^{\ell}\\
&\geq\frac{5}{24}\ell\cdot 16^{\ell}-\frac{1}{18}-\frac{2}{9}16^{\ell}\\
&\geq\frac{1}{24}\ell\cdot 16^{\ell} = \Omega(n^{2}\log n) & \MoveEqLeft \text{ using } \ell \geq 2 \qedhere
\end{align*}
\end{proof}

\bibliographystyle{plainurl}
\bibliography{biblio}

\end{document}